\pdfoutput=1

\documentclass[aps,pra,showkeys,amssymb,amsmath,amsfonts,twocolumn]{revtex4}
\usepackage{amsthm,graphicx}

\newcommand{\abs}[1]{\ensuremath{\left|#1\right|}}
\newcommand{\bra}[1]{\langle #1 |}
\newcommand{\ket}[1]{| #1 \rangle}
\newcommand{\ketbra}[2]{| #1 \rangle \langle #2 |}
\newcommand{\braket}[2]{\langle #1 | #2 \rangle}

\newcommand{\eg}{\emph{eg.}}

\newcommand{\Cplx}{\mathbb{C}}
\newcommand{\1}{{\rm 1\hspace{-0.9mm}l}}

\newtheorem{theorem}{Theorem}

\newtheorem{corollary}{Corollary}

\newcommand{\envelope}{(\raisebox{-.5pt}{\scalebox{1.45}{\Letter}}\kern-1.7pt)}

\begin{document}
\title{Increasing the security of the ping-pong protocol by using many mutually unbiased bases}

\author{Piotr Zawadzki}
\email{piotr.zawadzki@polsl.pl}
\affiliation{Institute of Electronics, Faculty of Automatic Control, Electronics
and Computer Science, Silesian University of Technology, Akademicka 16, 44-100
Gliwice}

\author{Zbigniew Pucha{\l}a}
\email{z.puchala@iitis.pl}
\author{Jaros{\l}aw Adam Miszczak}
\email{miszczak@iitis.pl} 
\affiliation{Institute of Theoretical and Applied Informatics, Polish Academy
of Sciences, Ba{\l}tycka 5, 44-100 Gliwice, Poland}

\begin{abstract}
In this paper we propose an extended version of the ping-pong protocol and study
its security. The proposed protocol incorporates the usage of mutually unbiased
bases in the control mode. We show that, by increasing the number of bases, it
is possible to improve the security of this protocol. We also provide the upper
bounds on eavesdropping average non-detection probability and propose a control
mode modification that increases the attack detection probability. 
\end{abstract}

\keywords{quantum cryptography, quantum secure direct communication, ping-pong protocol}
\pacs{03.67.Dd, 03.67.Hk, 03.65.Ud}

\maketitle

%
%
%
 
\date{28/02/2012 (v. 0.40)}

\section{Introduction}
A method of quantum secure direct communication (QSDC), contrary to quantum key
distribution (QKD) schemes, offers the confidential exchange of deterministic
messages without key agreement~\cite{korchenko10modern}. The interest in this fascinating idea started a
decade ago in seminal papers of~Beige~\emph{et.al.}~\cite{Beige-QSDC} and
Bostr\"om~\emph{et.al.}~\cite{Bostrom-pingpong-PhysRevLett.89.187902}. Since
then QSDC techniques have been developed following two different paradigms:
exploiting indistinguishability of non-orthogonal quantum
states~\cite{Beige-QSDC,ChinPhysLett.21.601,Luca-PhysRevLett.94.140501,OptCommun.283.1984,Wang2006256}
and based on entanglement of signal particles with a~system inaccessible to the
eavesdropper~\cite{Bostrom-pingpong-PhysRevLett.89.187902,Deng2006359,PhysRevA.80.022323,Zhan20094633}.
The protocols from the former family are usually simpler to implement at the
price of classic channel utilization in message mode, although exceptions of
this rule exist~\cite{Luca-PhysRevLett.94.140501}. On the other hand, the
entanglement based ping-pong protocol uses classic channel only in control
mode~\cite{Bostrom-pingpong-PhysRevLett.89.187902}. This feature can be
exploited to build and additional cryptographic security layer which improves
security of the protocol~\cite{vasiliu-pppa,pz:qinp-privacy-2012}. The ping-pong
protocol has been also improved and extended in other directions including
super-dense information
coding~\cite{Cai-pingpong-superdense-PhysRevA.69.054301,Wang-qudits} and its
variants based on higher dimensional signal
particles~\cite{Vasiliu-qutrits,zawadzki11security}. However, in the analyses of
higher dimensional variants it was assumed that control mode is executed in at
most two dual bases. This possibly understates an eavesdropping detectability.

The main aim of this paper is to show that, by increasing the number of bases
used in the control mode, it is possible to decrease an upper bound of the
attack non-detection probability. Eavesdropping is most effectively detected if
subsequent tests are executed in randomly selected mutually unbiased bases
(MUB)~\cite{durt10onmubs}. Unfortunately, the problem of finding MUB for the
arbitrary Hilbert space remains unsolved and constructive solutions exist only
for spaces of dimension $N=p^m$ where $p$ is
prime~\cite{durt2005mutually,eusebi09Deterministic} and/or spaces with dimension
not exceeding six~\cite{brierley-2010-10,mcnulty}.

\section{Preliminaries}

\subsection{Mutually Unbiased Bases}
A sequence of orthonormal bases $\{\mathcal{B}^{(0)},\mathcal{B}^{(1)}, \dots,
\mathcal{B}^{(M)}\}$ of $\Cplx^N$ is called MUB if, for any two elements
$\ket{b_k^{(m)}}\in\mathcal{B}^{(m)}$, $\ket{b_l^{(n)}}\in\mathcal{B}^{(n)}$,
the following condition holds
\begin{equation}\label{eq:MUBS-condition}
\abs{\braket{b_k^{(m)}}{b_l^{(n)}}}^2 = \delta_{m,n}\delta_{k,l} + \frac{1}{N} \left(1-\delta_{m,n}\right),
\end{equation}
where $N$ denotes the dimension of underlying Hilbert space. The explicit
construction of MUB is only known in the case of dimension $N=p^m$, where $p$
is a prime and $m$ is a positive integer~\cite{durt2005mutually}.
For an odd prime $p$ we have~\cite{durt2005mutually}
\begin{equation}\label{eq:MUBS-expression}
\begin{split}
\ket{b_k^{(l)}}
&= \sum\limits_{q=0}^{N-1} B^{(l)}_{k,q} \ket{b_q^{(0)}}
\\
&= \frac{1}{\sqrt{N}} \sum\limits_{q=0}^{N-1} \omega^{\ominus k \odot q} \omega^{(l-1)\odot q \odot q \oslash 2} \ket{b_q^{(0)}},
\end{split}
\end{equation}
where $\omega = e^{2\pi i/N}$, circled operations $\odot$, $\oslash$, $\ominus$
denote multiplication, division and subtraction in the finite field
$\mathrm{GF}(p^m)$ respectively and $\ket{b_q^{(0)}}$ are vectors of
computational basis. In the case of $p=2$ the explicit formulas for the MUB
elements are more involved~\cite{eusebi09Deterministic}.

\subsection{Ping-pong protocol operation}
Bob, the recipient of information,
prepares an EPR pair composed of qudits~\cite{Durt-pingpong-qudits-PhysRevA.69.032313}
\begin{equation}
\ket{\psi_{0,0}}=
\frac{1}{\sqrt{N}}\sum\limits_{k=0}^{N-1} \ket{b^{(0)}_k}\ket{b^{(0)}_k}.
\end{equation}
One of the qudits, referred to as 'home', is kept confidential, while the second
one, called 'travel', is sent to Alice. Because of an entanglement, Alice's
manipulations on the travel qudit induce non-local effects. Alice is able to
encode $2\log_2N$ bits of information per one protocol cycle applying one of the
unitary transformations
\begin{equation}\label{pz:eq:encoding-operation}
U_{\mu,\nu} = \sum\limits_{k=0}^{N-1} \omega^{\mu k} \ketbra{b^{(0)}_{k+\nu}}{b^{(0)}_{k}},
\end{equation}
where $\mu,\nu=0, \dots, N-1$. Operator~\eqref{pz:eq:encoding-operation}
transforms the initial state into another EPR pair
$\ket{\psi_{\mu,\nu}}$~\cite{Liu-Qudit-superdense-coding-PhysRevA.65.022304}
which can be unambiguously discriminated by Bob when the 'travel' qudit is
returned by Alice. Eavesdropping Eve cannot distinguish the~travel qudit on its
way forth and back from a maximally mixed state
\begin{equation}\label{pz:eq:travel-mixture}
\rho_t = \frac{1}{N} \sum\limits_{\alpha=0}^{N-1} \ketbra{b^{(0)}_\alpha}{b^{(0)}_\alpha},
\end{equation}
so this way she cannot infer any information about the encoding operation used
by Alice. Because of that indistinguishability, further analysis can be carried
out as if Bob sent one of the~randomly selected
states~$\ket{b^{(0)}_\alpha}$~\cite{Bostrom-pingpong-PhysRevLett.89.187902}.
However, Eve can entangle the 'travel' qudit with some ancilla system before it
reaches Alice
\begin{equation}\label{pz:eq:noncoherent-attack}
\ket{\psi_\alpha} = A\ket{b^{(0)}_\alpha,\phi} = \sum\limits_{l=0}^{N-1} a_{\alpha,l}\ket{b^{(0)}_l,\phi_{\alpha,l}},
\end{equation}
where $\alpha=0,\ldots,N-1$ and $\ket{\phi_{\alpha,l}}$ denotes Eve's probe
states. That way, because of the introduced entanglement, Alice's encoding
operation also modifies the state of the ancilla. By inspection of the ancilla's
state Eve can gain some information about the encoded message. On the other
hand, Eve's attack operation inevitably breaks the perfect correlation of the
'travel' and 'home' qudits, and that violation can be detected when Alice and
Bob switch to control mode in which they perform local measurements on the
possessed qudits and classically communicate their results. Unfortunately, the
control mode executed only in computational basis is insufficient, as Eve can
mount an undetectable attack in which she can infer half the information posted
by Alice~\cite{Vasiliu-qutrits,zawadzki11security}. It has been also shown
in~\cite{zawadzki11security} that the incorporation of dual basis removes such
possibility. The question how protocol detectability can be improved by taking
into account all possible mutually unbiased bases remains open.

Without loss of generality it may be assumed that 
Bob sends a state $\ket{\alpha}$~\cite{Bostrom-pingpong-PhysRevLett.89.187902,Vasiliu-qutrits}.
It follows from \eqref{pz:eq:noncoherent-attack} that
$p^{(0)}_\alpha=\abs{a_{\alpha,\alpha}}^2$ describes the non-detection probability when
computational basis $\mathcal{B}^{(0)}$ is used in control mode.
If Alice selects another basis $\mathcal{B}^{(m)}$ then the 'travel' qudit after attack
is seen as
\begin{equation}\label{pz:eq:noncoherent-attack-basis}
\ket{\psi_\alpha} = A\ket{b^{(0)}_\alpha,\phi} 
=
\sum\limits_{k=0}^{N-1}
c_{\alpha,k}
\ket{b^{(m)}_k,\phi_{\alpha,l}},
\end{equation}
where 
$c_{\alpha,k} = \sum\limits_{l=0}^{N-1} a_{\alpha,l} \braket{b^{(m)}_k}{b^{(0)}_l}$.
The attack is not detected in the basis $\mathcal{B}^{(m)}$ with probability
\begin{equation}
p^{(m)}_\alpha = \abs{c_{\alpha,\alpha}}^2
= \abs{\braket{b^{(m)}_\alpha}{a_{\alpha,:}}}^2,
\end{equation}
where $\ket{a_{\alpha,:}}=\sum\limits_{l=0}^{N-1} a_{\alpha,l} \ket{b^{(0)}_l}$.
The non-detection probability averaged over multiple control mode cycles
is given by
\begin{equation}
\label{eq:non-detection-def}
d_\alpha = \sum\limits_{m=0}^{M-1} q_m p^{(m)}_\alpha,
\end{equation}
where $M$ is the number of bases and $q_m$ describes relative frequency of their
selection.

It should be shown for completeness that in the control mode
Bob can unambiguously infer Alice's local measurement result
as long as he is informed about the used basis.
This follows from the fact that, as the local change of basis does not influence
the entanglement, the measurement performed by Alice fully determines the
outcome of Bob's measurement.
Let us suppose that Alice performed a measurement in the basis $\mathcal{B}$ and
obtained symbol $i$. In this case, the state of the system, after the projective
measurement, reads 
\begin{equation}
\begin{split}
\left(\ketbra{U_i}{U_i} \otimes \1 \right) 
\left( \sum_{k} \ket{k} \otimes \ket{k} \right) 
&=
\ket{U_i} \otimes \left(\sum_{k} \braket{U_i}{k}\ket{k} \right) \\
&=
\ket{U_i} \otimes \ket{\overline{U_i}},
\end{split}
\end{equation}
where $U_i$ is $i^{\text{th}}$ vector of the basis $\mathcal{B}$. From the above
one can notice that, if Bob performs a measurement in the basis
$\overline{\mathcal{B}}$, he will obtain the symbol $i$ with probability~1.

\section{Bounds on the non-detection probability}
Let us begin with general theorem concerning non-detection probability.
\begin{theorem}\label{th:max-prob-svd}
Let $\{\mathcal{B}^{(0)},\mathcal{B}^{(1)}, \dots, \mathcal{B}^{(M)}\}$ be a set
of $M+1$ orthonormal bases, used in the control mode of the protocol and
selected equally frequently. Then, the upper bound for the average non-detection
probability~\eqref{eq:non-detection-def} is given by
\begin{equation}\label{eq:dmax-singular}
d_\alpha \le \frac{1}{M+1} \sigma_1^2(V^{(\alpha)}),
\end{equation}
where $\sigma_1(V^{(\alpha)})$ denotes the greatest singular value of 
 $V^{(\alpha)}= \left\{\overline{B}^{(i)}_{\alpha,j} \right\}_{ij}$.
\end{theorem}
\begin{proof}
Let us denote by $v_m$ the $\alpha^{\text{th}}$ element of $\mathcal{B}^{(m)}$. 
By $V^{(\alpha)}$ we denote a matrix with rows given by bra vectors $\bra{v_m}$,
\emph{i.e.} $V^{(\alpha)}_{m,j}
 = \overline{B}^{(m)}_{\alpha,j}$ (overline denotes complex conjugate).
If control bases are selected equally frequently
the average non-detection probability~\eqref{eq:non-detection-def} can
be written as
\begin{equation}\label{eq:non-detection-matrix}
d_{\alpha} = \frac{1}{M+1} \left\| V^{(\alpha)} \ket{a_{\alpha,:}}\right\|^2,
\end{equation}
and since $\max_{\ket{x}}\left\| V^{(\alpha)} \ket{x}\right\|^2 = \sigma_1^2(V^{(\alpha)})$
we obtain the result.
\end{proof}

Let us now assume that the control mode is executed in $M+1$ mutually unbiased
bases. In this case the upper bound on the non-detection probability is stated
in the following theorem.
\begin{theorem}\label{th:M-MUBs-bound}
If the control mode is executed in $M+1$ mutually unbiased bases, then the 
average non-detection probability is bounded by
\begin{equation}\label{eq:M-MUBs-bound}
d_{\alpha} \leq \frac{1+M/\sqrt{N}}{1 + M}.
\end{equation}
\end{theorem}
\begin{proof}
Let us introduce a matrix $W =V^{(\alpha)} V^{(\alpha)}{}^{\dagger}$
where $V^{(\alpha)}$ is defined as in the proof of Theorem~\ref{th:max-prob-svd}.
Directly from the definition of matrix $V^{(\alpha)}$
and MUB condition~\eqref{eq:MUBS-condition}
we get
\begin{equation}\label{eqn:elements_w}
W_{i,j} = \braket{v_i}{v_j}=\{\delta_{i,j} + (1-\delta_{i,j}) e^{\mathrm{i} \phi_{i,j}}/\sqrt{N} \}_{i,j=0}^{M}.
\end{equation}
Note that matrix $W$ does not depend on the particular $\alpha$. The maximal
singular value of the matrix $W$ can be bounded as \cite{schur11bemerkungen}
(see also inequality  \cite[Eq. (3.7.2)]{horn91topics})
\begin{equation}\label{eqn:schur-bound}
\sigma_1(W)  \leq \left((\max_{i} \sum_j|W_{i,j}|) (\max_{j} \sum_i|W_{i,j}|)\right)^{1/2}.
\end{equation}
Taking into account \eqref{eqn:elements_w} we get
\begin{equation}
\max_{i} \sum_j|W_{i,j}| = \max_{j} \sum_i|W_{i,j}| = 1+M/\sqrt{N},
\end{equation}
and the result follows from Theorem~\ref{th:max-prob-svd} and the fact that
$\sigma_1^2(V^{(\alpha)})=\sigma_1(W)$.
\end{proof}

In the case of dimension $N = p^m$ for prime $p$ and $m$ being a positive integer,
there exists a set of $N+1$ mutually unbiased bases~\cite{durt2005mutually}, and
the bound \eqref{eq:M-MUBs-bound} reads
\begin{equation}\label{eq:N-MUBs-bound}
d_{\alpha} \leq \frac{1+\sqrt{N}}{1 + N}.
\end{equation}
It is possible to improve this bound using explicit
expression~\eqref{eq:MUBS-expression}.
\begin{theorem} \label{th:odd-prime^m}
Let $N = p^m$ where $p$ is an odd prime and $m$ is a positive integer.
Then the maximal non-detection probability is bounded by 
\begin{equation}\label{eq:N-odd-MUBs-bound}
d_{\alpha} \le \frac{3}{1+N}.
\end{equation}
\end{theorem}
\begin{proof}
Let us introduce matrix $W=V^{(\alpha)}{}^{\dagger} V^{(\alpha)}$ for some fixed $\alpha$
\begin{equation}
W_{\mu,\nu} 
= \sum\limits_{q=0}^{N} (V^{(\alpha)}{}^\dagger)_{\mu,q} V_{q,\nu}
= \sum\limits_{q=0}^{N} B^{(q)}_{\alpha,\mu} \overline{B}^{(q)}_{\alpha,\nu}.
\end{equation}
Matrix $W$ may be decomposed as $W=P+Q$:
\begin{align}
P_{\mu,\nu} & = B^{(0)}_{\alpha,\mu} \overline{B}^{(0)}_{\alpha,\nu} = \delta_{\alpha,\mu}\delta_{\alpha,\nu} , \\
Q_{\mu,\nu} & = \sum\limits_{q=1}^{N} B^{(q)}_{\alpha,\mu} \overline{B}^{(q)}_{\alpha,\nu}
\nonumber \\ &
= \omega^{\ominus \alpha \odot \left( \mu \ominus \nu \right)}
\frac{1}{N} \sum\limits_{q=1}^{N} \omega^{ (q-1)\odot \left(\mu \odot \mu \ominus \nu \odot \nu \right) \oslash 2}
\nonumber \\ &
= \omega^{\ominus \alpha \odot \left( \mu \ominus \nu \right)} \delta_{\mu \odot \mu \ominus \nu \odot \nu,0},
\end{align}
where we have used identities $\omega^k \omega^l = \omega^{k \oplus l}$ and
$\sum_{k=0}^{N-1} \omega^{k \odot l} = N \delta_{l,0}$
(see \eg{}~\cite{durt2005mutually}). Thus, $\abs{Q_{\mu,\nu}}=\delta_{(\mu \ominus
\nu)\odot(\mu\oplus\nu),0}$ and using bound \eqref{eqn:schur-bound} we get
$\sigma_1(Q)\le 2$. Obviously $\sigma_1(P)=1$. Thesis follows
from~\eqref{eq:dmax-singular} combined with $\sigma_1^2(V)=\sigma_1(W)$ and
inequality~\cite[Eq. (3.3.17)]{horn91topics}
\begin{equation}
\sigma_1(P+Q) \le \sigma_1(P) + \sigma_1(Q) = 3 .
\end{equation}
\end{proof}

It has been shown that usage of at least two dual bases is sufficient to ensure
asymptotic protocol security~\cite{zawadzki11security}. However, the detection
capabilities of the control mode are significantly improved when more mutually
unbiased bases are used. This follows from the comparison of the bound on
average non-detection probability obtained in~\cite{zawadzki11security} (curve
(d) on Fig.~\ref{fig:nondetection}) with the bounds obtained herein (curves (b)
to~(d)). The improvement of the protocol's detectability becomes more apparent
with the increase of the dimension of the underlying Hilbert space~-- for bound
from~\cite{zawadzki11security} we have $\lim_{N\to\infty} d_{\rm max}=1/2$ while
for~\eqref{eq:N-MUBs-bound} and~\eqref{eq:N-odd-MUBs-bound} $\lim_{N\to\infty}
d_{\rm max}=0$. Although the asymptotic behavior of the bounds
\eqref{eq:N-MUBs-bound} and~\eqref{eq:N-odd-MUBs-bound} is similar, they differ
in the provided optimality.

The comparison of the considered bounds with numerical results is~presented
in~Fig.~\ref{fig:nondetection}. It follows that bound
\eqref{eq:N-odd-MUBs-bound} is close to optimal. It was also verified that the
best fitting to numerical estimates is achieved for $\sigma^2_1(V)=\frac{1}{2}(3
+ \sqrt{5}) \approx 2.618$.

\begin{figure}
\centering
    \includegraphics[width=\columnwidth]{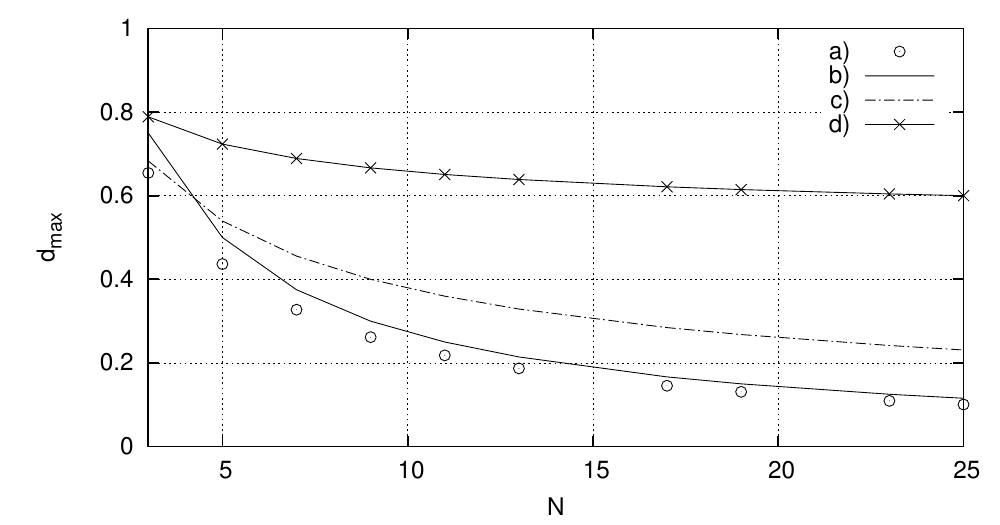}
    \caption{Comparison of upper bounds on average probability of non-detection
    calculated: a) via numerical simulations, b) with expression
    \eqref{eq:N-odd-MUBs-bound}, c) with expression \eqref{eq:N-MUBs-bound}, d)
    when only two bases are used in control mode~\cite{zawadzki11security}
    ($d_{\rm max}=(1+1/\sqrt{N})/2$).
\label{fig:nondetection}}
\end{figure}

Further improvement can be proposed based on the analysis of the proof of 
Theorem~\ref{th:odd-prime^m}. The matrix $P$ is related to the control mode
tests executed in the computational basis. If that basis is excluded from the
control mode, one obtains a better protocol behaviour. We can state the
following.

\begin{corollary}\label{th:N-odd-nc-MUBs-bound}
Let us assume, that $N=p^m$, where $p$ is an odd prime and $m$ is a positive
integer, and the computational basis is excluded from the control mode. In this
case an average non-detection probability is bounded by
\begin{equation}\label{eq:N-odd-nc-MUBs-bound}
d_{\alpha} \leq 2/N.
\end{equation} 
\end{corollary}

It should be noted that, as the information is encoded and decoded in the
computational basis, Eve still has to 
use this basis for an attack preparation.
Comparing the above bound with the numerical estimate for the seminal protocol,
we observe about $30\%$ improvement in an attack detection
capabilities. 

\section{Conclusions}
In this paper we have proposed an extended version of the ping-pong protocol, which 
incorporates the usage of mutually unbiased bases in the control mode. We
provided upper bounds on eavesdropping average non-detection probability in the
proposed protocol.

If the communicating parties use $M+1$ mutually unbiased
bases in the control mode, the bound is given by the leading singular value of 
the matrix with rows given by the appropriate bra vectors.
One should note that the number $M$ of bases used in the control mode should
depend on the dimension.

If the communicating parties use particles of dimension $N = p^m$, where  $p$ is
an odd prime and $m$ is a positive integer, it is possible to provide a better
estimate. Assuming that Alice and Bob use $N+1$ bases and construct them
according to~\cite{durt2005mutually}, the non-detection probability averaged
over sufficiently many cycles never exceeds $3/(N+1)$. Eavesdropping detection
capabilities can be improved by the exclusion of the computational basis from
the control mode.

\begin{acknowledgments} 
This work was supported by the Polish National Science Centre under the research
project N~N516~475440. The authors would like to thank K.~\.Zyczkowski and P.~Gawron
for interesting discussions.
\end{acknowledgments}

\end{document}